\documentclass{l4dc2026}

\usepackage{natbib}

\usepackage{times}
\usepackage{thmtools}
\usepackage[capitalise,noabbrev]{cleveref}
\crefname{lemma}{lemma}{lemmas}
\Crefname{lemma}{Lemma}{Lemmas}
\usepackage{bbm}
\usepackage{bm} 
\usepackage{graphicx}
\usepackage{amsmath}

\usepackage{booktabs}       % professional-quality tables
\usepackage{multirow} 
\title[Probabilistic Safety Guarantee]{Probabilistic Safety Guarantee for Stochastic Control Systems Using Average Reward MDPs}

\author{%
 \Name{Saber Omidi} \Email{saber.omidi@unh.edu}\\
 \addr Department of Mechanical Engineering\\
       University of New Hampshire\\
       33 Academic Way
       Durham, NH 03824\\
 \AND
 \Name{Marek Petrik} \Email{marek.petrik@unh.edu}\\
 \addr Department of Computer Science\\
       University of New Hampshire\\
       33 Academic Way
       Durham, NH 03824\\
 \AND 
 \Name{Se Young Yoon} \Email{seyoung.yoon@unh.edu}\\
 \addr Electrical and Computer Engineering\\
       University of New Hampshire\\
       33 Academic Way
       Durham, NH  03824 \\
 \AND
 \Name{Momotaz Begum} \Email{momotaz.begum@unh.edu}\\
 \addr Department of Computer Science\\
       University of New Hampshire\\
       33 Academic Way
       Durham, NH 03824
}

\definecolor{NavyBlue}{rgb}{0.0, 0.0, 0.5} % A dark blue

\newcommand{\PiSD}{\Pi_{\mathrm{SD}}}

\newcommand{\Real}{\mathbb{R}}
\newcommand{\Nats}{\mathbb{N}}
\renewcommand{\P}{\mathbb{P}}

\newcommand{\opt}{^\star}
\DeclareMathOperator*{\argmax}{arg\,max}

\begin{document}

\maketitle

\begin{abstract}%
Safety in stochastic control systems, which are subject to random noise with a known probability distribution, aims to compute policies that satisfy predefined operational constraints with high confidence throughout the uncertain evolution of the state variables. The unpredictable evolution of state variables poses a significant challenge for meeting predefined constraints using various control methods. To address this, we present a new algorithm that computes safe policies to determine the safety level across a finite state set. This algorithm reduces the safety objective to the standard average reward Markov Decision Process (MDP) objective. This reduction enables us to use standard techniques, such as linear programs, to compute and analyze safe policies. We validate the proposed method numerically on the Double Integrator and the Inverted Pendulum systems. Results indicate that the average-reward MDPs solution is more comprehensive, converges faster, and offers higher quality compared to the minimum discounted-reward solution.
\end{abstract}


\section{Introduction}\label{sec:Intro}

Safety-critical algorithms are a vital requirement for stochastic control systems deployed in fields such as autonomous robots, quadrotors, and self-driving cars to prevent injuries or financial losses~\citep{dawood2024dynamic, zhong2025bridging, soleimani2025safe}. Traditionally, safety problems are often framed as a reach-avoid problem, where the agent is guaranteed to reach the intended goal while avoiding actions that may lead to undesirable states~\citep{compton2024learning,rabiee2023safe,wabersich2021predictive,alan2023control}. These approaches achieve safety by treating the domain's stochasticity using a worst-case adversarial, or robust, approach that can tractably compute safe policies by solving the Hamilton-Jacobi-Isaacs (HJI) equation~\citep{chen2017robust, margellos2011hamilton, moon2022state, arnstrom2024data}. Although these standard approaches have been used successfully in some domains, they often struggle in domains with significant uncertainty~\citep{dallas2025control}.    

In this paper, we propose a new approach to computing safe policies in stochastic control by reducing the problem to the average-reward criterion in Markov Decision Processes~(MDPs)~\citep{puterman1990markov}. MDPs represent a flexible framework that is used to model reinforcement learning problems. In particular, we make the following contributions in this paper. (1) We show that average-reward MDPs provide a direct and computationally efficient approach for calculating the probabilistic safety value function using linear programs. (2) Our average-reward-based approach uses standard tools to dispense with the need for artificial discount factors used in prior work~\cite{akametalu2023minimum}. Previously, discount factors were used to improve the computational complexity of computing the fixed point. (3) Our reduction makes it possible to compute states that are safe with high confidence, rather than just computing safe and unsafe states. 
 
Although there is a rich literature on safety, the analysis of probabilistic safety remains insufficiently addressed, particularly when robust control methods are overly conservative or when worst-case scenarios cannot be precisely characterized. Overly conservative control policies significantly constrain the operational range of real-world stochastic systems, where high-risk events are infrequent. Moreover, MDPs commonly employ a discount factor in safety analyses; however, even average reward frameworks for reach-avoid problems do not guarantee probabilistic safety in the system's long-term behavior. 

Our approach departs significantly from common approaches to safety in stochastic control~\cite{chen2017robust}. Most existing literature on the subject guarantees safety by reducing the stochastic outcomes to adversarial noise. The adversarial approach enables the use of deterministic differential game theory tools to formulate and solve the safety problem~\citep{margellos2011hamilton, moon2022state, arnstrom2024data}. Central to this framework is the solution of the Hamilton-Jacobi-Isaacs (HJI) equation~\citep{evans1984differential, abu2006policy, akametalu2023minimum, begzadic2025back}. For instance, \cite{fisac2018general} introduces an HJI-based safety framework as a variational inequality, guaranteeing constraint satisfaction while minimizing disruption to the learning process. This framework also incorporates a Bayesian mechanism to adaptively update the safety analysis as new data is collected. 

In related work, \cite{avila2021reachability} also reduces infinite-horizon reach-avoid problems to the average-reward criterion for MDPs. Similarly to our approach, they probabilistically characterize reachability and reach-avoid solutions, inspired by the stochastic target problem~\citep{soner2002stochastic}. However, our safety objective cannot be immediately cast as a reach-avoid problem. Our work is also related to \cite{gao2023distributional}, which addresses the limitations of traditional MDP reachability in explicitly accounting for transient distributions.

The remainder of the paper is organized as follows. \cref{sec:pro_state} introduces the preliminaries for safety in uncertain control. \cref{sec:MDP}  presents the Average Reward MDP criterion for the probabilistic safety problem. \cref{sec:r-work} reviews related work in the HJI framework. \cref{sec:N_results} presents a numerical validation, and \cref{sec:Conc} concludes the paper.

\section{Problem Statement} \label{sec:pro_state}

The analysis of safety-critical systems requires a formal method to guarantee the satisfaction of state and input constraints even when the system is subject to stochastic disturbances. This section presents the foundational framework for providing probabilistic safety guarantees in control systems affected by disturbances. We seek to identify a set of safe states that account for the system’s uncertain response to control actions.

\textbf{Notation:} Vectors and matrices are denoted by bold lowercase and uppercase letters, respectively. Random variables are distinguished by the use of tildes. Subscripts generally indicate time. The symbols $\mathbb{R}$, $\Real_+$, and $\mathbb{N}$ denote the sets of real, non-negative real, and natural numbers, respectively. We use $\Delta_n \subseteq \mathbb{R}^n$ for $n \in \mathbb{N}$ to represent the probability simplex.
To this end, we assume a probability space $(\Omega, \mathcal{B}, p)$ and we use $\mathbb{X}_{\mathcal{S}}$ to denote the set of all $\mathcal{S}$-valued random variables.

For finite and non-empty state set $\mathcal{S} = \{1, 2, 3, \dots, S\}$ and action set $\mathcal{A} = \{1, 2, 3, \dots, A\}$, denote the state and action variables at time $k \in \mathbb{N}$ as $s_k \in \mathcal{S}$ and $a_k \in \mathcal{A}$, respectively.
The evolution of the state over time is influenced by an external disturbance $d$, which takes values in the set $\mathcal{D} = \{1, 2, 3, \dots, D\}$. 
The disturbance function $d \colon \mathcal{S} \to \mathcal{D}$ depends on the current state. 
The system dynamics are described by a known difference equation, where the function $f \colon \mathcal{S} \times \mathcal{A} \times \mathcal{D} \to \mathcal{S}$ maps $s_k$, $a_k$, and $d(s_k)$ to the next state $s_{k+1}$:
\begin{align}\label{eq:dif_tranistion}
s_{k+1} = f(s_k, a_k, d(s_k)).
\end{align}
A random disturbance is represented by a random variable $\tilde{d} \colon \mathcal{S} \to \mathbb{X}_{\mathcal{D}}$ where
$\mathbb{X}_{\mathcal{D}}$ is 
the set of all $\mathcal{D}$-valued, $\mathcal{B}$-measurable random variables. 
Let $\pi \colon \mathcal{S} \to \mathcal{A}$ denote a deterministic policy that prescribes a control action for each state, and define the set of all stationary deterministic policies as $\PiSD := \mathcal{A}^{\mathcal{S}}$. 
For each policy $\pi \in \PiSD$, we define a stochastic process $\tilde{x}_{k} \colon \mathcal{S} \to \mathbb{X}_{\mathcal{S}}$ where the system dynamics in~\eqref{eq:dif_tranistion} are reformulated as a stochastic process starting from a state $s \in \mathcal{S}$, and the probability of the state transition equals,
\begin{align}\label{eq:stochastic_process}
  \mathbb{P}_s^{\pi} \left[\tilde{x}_{k+1} = f\big(\tilde{x}_k, \pi(\tilde{x}_k),\, \tilde{d}(\tilde{x}_k)\big), \hspace{2pt} \forall k \in \mathbb{N}\right] = 1,  \qquad
  \mathbb{P}_s^{\pi} \left[\tilde{x}_0 = s \right] = 1.
\end{align}
Here, $\mathbb{P}_{s}^{\pi}$ represents the probability measure over the system's future state trajectories, given that it starts in state $s$ and follows a specific policy $\pi$.
The random variables defined in \cref{eq:stochastic_process} as $\tilde{x}_k$ are measurable with respect to the standard filtration of the stochastic process~\citep{jazwinski2007stochastic}.  
After defining \cref{eq:stochastic_process}, we now provide definitions of the constraint and safe sets.
\begin{definition}\label{def:cons_set}
The constraint set $\mathcal{C} \subseteq \mathcal{S}$ represents admissible states.
\end{definition}
While $\mathcal{C}$ defines all instantaneously admissible states, we seek a subset of $\mathcal{C}$ from which a policy can guarantee evolution of \cref{eq:stochastic_process} remains within $\mathcal{C}$ indefinitely with a certain level of probability.

\begin{definition}\label{def:safe_set_1}
The \emph{probabilistically-safe} set $\mathcal{K}_{\alpha}\subseteq \mathcal{S}$ with confidence $\alpha\in [0,1]$ comprises states for which there exists a $\pi\in \PiSD$ that guarantees that the process remains in the constraint set $\mathcal{C}$ with probability of $\alpha$: 
\begin{equation}
\label{def:safe_set_2}
\mathcal{K}_{\alpha} := \left\{s \in \mathcal{S} \mid \exists \pi\in \PiSD, \P^{\pi}_s \left[\tilde{x}_k \in \mathcal{C}, \forall k\in \Nats\right] \ge \alpha  \right\}.
\end{equation}
Let $\mathcal{K}:=\mathcal{K}_1$ be the safe set.
\end{definition}
It is easy to see that 
\[
  1 \ge \alpha_1 \ge \alpha_2 \ge 0
  \quad\implies\quad
  \mathcal{K}_{\alpha_1} \subseteq \mathcal{K}_{\alpha_2}.
\]
The notation $\mathbb{P}^{\pi}_s \left[\tilde{x}_k \in \mathcal{C}, \forall k\in \mathbb{N}\right]$ specifies that the entire future trajectory evolving from $\tilde{x}_k$ under policy $\pi$ must remain within $\mathcal{C}$, ensuring safety over the entire infinite time horizon. The objective of the probabilistic safety problem, therefore, is to compute the probabilistic safe set $\mathcal{K}_{\alpha}$ and then find a policy $\pi \in \Pi_{SD}$ that can gurantee the system evolution according to \cref{eq:stochastic_process} remains in $\mathcal{C}$ with $\alpha$ level confidence.

\section{Average Reward MDPs for Probabilistic Safety} \label{sec:MDP}
This section addresses the probabilistic safety problem within the multichain average reward MDP. First, we define the components and properties of the specialized MDP model. Second, we provide lemmas and proofs to formally connect the average reward framework to the probabilistic safety problem. Finally, we establish the connection between average reward MDPs and linear programs.

Consider the MDP $(\mathcal{S}, \mathcal{A}, p, r, \mu)$ where 
$\mathcal{S}$ and 
$\mathcal{A}$ are finite state and action sets, respectively.  
The function \( p \colon \mathcal{S} \times \mathcal{A} \times \mathcal{S} \to [0,1] \) is the \emph{transition-probability function}, where $p(s,a,s')$ indicates the probability of transitioning to the next state $s'$ from the current state $s$ under action $a$. We define the transitions as 
\begin{equation} \label{eq:transition-function}
p(s, a, s') :=
\begin{cases}
\P\left[s'=f(s, a, \tilde{d}(s)) \right] & \text{if } s\in \mathcal{C}, \\
\mathbbm{1}\{ s' = s \} & \text{otherwise},
\end{cases}
\end{equation}
where $\mathbbm{1}$ denotes the indicator function.
If the MDPs state $s$ is admissible ($s \in \mathcal{C}$), the next state $s'$ probabilistically follows the outcome of the system dynamics, $f(s, a, \tilde{d}(s))$, as described in \cref{eq:dif_tranistion}. Conversely, any occurrence of a violation ($s_k \notin \mathcal{C}$) results in all future states remaining inadmissible (i.e., they become absorbing states). 
The motivation for the construction in \cref{eq:transition-function} is to identify safety violations in the entire horizon by system evolution in \cref{eq:stochastic_process}, as there is a possibility that the system evolution enters the inadmissible states and returns to the admissible state in the infinite horizon.  

The \emph{reward function} $r \colon \mathcal{S} \times \mathcal{A} \times \mathcal{S}  \to \mathbb{R}$ represents the immediate reward. In particular, $r(s,a,s')$ is the reward received when taking an action $a$ in state $s$ and transitioning to $s'$. We define the immediate reward function to reflect the safety objective as follows:
\begin{align}\label{eq:reward_indicator_function}
  r(s,a,s') \;:=\mathbbm{1}\{s \in \mathcal{C}\}.
\end{align}
Definition in \cref{eq:reward_indicator_function} means the reward is 1 when the system is in a safe state and 0 otherwise, and the definition of the reward function in \cref{eq:reward_indicator_function} will encourage the average reward MDP system to remain within the admissible state region. 
An initial distribution over $\mathcal{S}$ is represented by $\mu \in \Delta_{\mathcal{S}}$, and it is assumed that $\mu(s) > 0$ for all $s \in \mathcal{S}$.

The objective for the \emph{ average reward criterion} is:
\begin{align} \label{eq:main_objective}
    \sup_{\pi \in \PiSD}
    \limsup_{N \to \infty} \, \frac{1}{N} \, \mathbb{E}^{\pi}_{\mu} 
    \left[
        \sum_{k=0}^{N-1} \mathbbm{1}\{\tilde{s}_k \in \mathcal{C}\}
    \right],
\end{align}
where the superscript $\pi$ denotes a policy from the stationary policy set $\Pi_{SD}$. Random variable 
$\tilde{s}_k$ represents the state of the MDP system at time $k$ \footnote{To be clear in this paper, $\tilde{x}_k$ refers to the state trajectory evolution according to \cref{def:safe_set_2} while $\tilde{s}_k$ is the trajectory of the MDP model when \cref{eq:transition-function} defined to formulate the probabilistic safety problem.}.
For the average reward criterion, the \emph{gain} function (known as the average-reward value function) $g: \mathcal{S} \to \mathbb{R}$ and the \emph{bias} function $h: \mathcal{S} \to \mathbb{R}$ are defined as:
\begin{align}
    \label{eq:av_MDP_gain_bias}
    g^\pi(s) 
    \!=\! \limsup_{N \to \infty} \frac{1}{N} \, \mathbb{E}^{\pi}_{\mu}\!
    \left[
        \sum_{k=0}^{N-1}\! \mathbbm{1}\{\tilde{s}_k \in \mathcal{C} \}
   \! \right],
    \quad
    h^{\pi}(s) 
    \!=\! \limsup_{N \to \infty} \mathbb{E}^{\pi}_{\mu}\!
    \left[
        \sum_{k=0}^{N-1}\! \left(\mathbbm{1}\{\tilde{s}_k \in \mathcal{C} \} \!-\! g^\pi(\tilde{s}_{k})\! \right)
    \right].
\end{align}
The optimal gain and the optimal policy are:
\begin{align}\label{eq:opt_policy}
g\opt = \max_{\pi \in \PiSD} g^\pi(s), \quad 
    \pi\opt \in \argmax_{\pi \in \PiSD} g^\pi(s).
\end{align}
The gain function in \cref{eq:av_MDP_gain_bias} is equivalent to the probabilistic safety value function, a key identity that tightly couples the system's long-run behavior to this value; as $N \to \infty$, this gain correctly becomes zero upon a state constraint violation. Within the multichain average-reward framework, the bias term is also necessary to satisfy the optimality equations, capturing the system's transient behavior while the gain captures its long-term behavior. The fundamental equations formally define gain and bias, connecting them to the system's underlying transitions~\citep[Theorem 8.2.6]{puterman1990markov}. Moreover, the multichain formulation is necessary because the optimal stationary policy may not be unique in \cref{eq:opt_policy}, leading to multiple recurrent classes~\citep[Section 9.1.3]{puterman1990markov}.

The following theorem characterizes the probabilistic safe set in terms of the optimal gain function, establishing a fundamental link between probabilistic safety and long-run average rewards under optimal policies.

\begin{theorem} \label{theorem:alpha_level}
For every state \(s \in S\) and confidence level \(\alpha \in [0,1]\):
\[
   s \in \mathcal{K}_\alpha
   \;\Longleftrightarrow\;
   \alpha \;\le\; g\opt (s).
\]
\end{theorem}
\begin{proof}
Suppose that $s\in \mathcal{K}_\alpha$.
There exists a policy $\hat{\pi} \in \PiSD$ such that for each $N\in \Nats$
\begin{equation*}
\alpha
\le \mathbb{P}^{\hat{\pi}}_s \left[\tilde{x}_k \in \mathcal{C},\;\forall k\in\mathbb N \right] 
  = \mathbb{P}^{\hat{\pi}}_s \left[\tilde{s}_k \in \mathcal{C},\;\forall k\in\mathbb N \right]
  \le \mathbb{P}^{\hat{\pi}}_s \left[\tilde{s}_k \in \mathcal{C},\;\forall k = 0, \dots , N-1 \right] 
\end{equation*}
Therefore, by the definition of gain and the fact that all rewards are non-negative, we get that
\begin{align*}
  g\opt(s)
  &\ge g^{\hat{\pi}}(s) = \limsup_{N \to \infty}\frac{1}{N}\,\mathbb{E}^{\hat{\pi}}_{\mu} 
  \left[\,\sum_{k=0}^{N-1} \mathbbm{1}\{\tilde{s}_k \in \mathcal{C}\}\right]   \\ 
&\ge \limsup_{N \to \infty}\frac{1}{N}  N \cdot \mathbb{P}^{\hat{\pi}}_s \left[\tilde{s}_k \in \mathcal{C},\;\forall k = 0, \dots , N-1 \right] 
\ge \limsup_{N \to \infty}\frac{1}{N}  N \cdot \alpha = \alpha.
\end{align*}

Now consider the optimal gain and summation of the Bellman equation over the entire $\mathcal{S}$: 
    \begin{align*}
          g\opt(s)
          = \sum_{s' \in \mathcal{K}} \mathbb P^{\pi\opt}_s \left[\tilde{s}_k=s'\right]g\opt(s') 
          + \sum_{
          s' \in \mathcal{C}\setminus \mathcal{K}} \mathbb P^{\pi\opt}_s \left[\tilde{s}_k=s'\right]g\opt(s') 
          + \sum_{s' \in \mathcal{S} \setminus\mathcal{C}} \mathbb P^{\pi\opt}_s \left[\tilde{s}_k=s'\right]g\opt(s').
    \end{align*} 
According to 
 Lemma \ref{lem:gain-failed} and Lemma \ref{lem:gain-one-safe} in 
\cref{app:proof-thm3},
we know that $g\opt (s') = 0$ if $s' \in \mathcal{S} \setminus \mathcal{C}$, and $g\opt (s') = 1$ if $s' \in \mathcal{K}$. Therefore,
\begin{align*}
          g\opt(s)
          =  \sum_{s' \in \mathcal{K}} \mathbb P^{\pi\opt}_s \bigl[\tilde{s}_k=s'\bigr] \cdot 1
          +\sum_{s' \in \mathcal{C}\setminus \mathcal{K}} \mathbb P^{\pi\opt}_s \bigl[\tilde{s}_k=s'\bigr]g\opt(s'). 
\end{align*}
Since $\mathcal{C} \setminus \mathcal{K}$ is a transient set based on Lemma \ref{lem:transient_state} (see \cref{app:proof-thm3}), we have 
\(
\lim_{k \to \infty} \mathbb{P}^{\pi\opt }_s\!\left[\tilde{s}_k = s'\right] = 0, \quad \forall s' \in \mathcal{C} \setminus \mathcal{K}.
\)
Taking the limit:
    \[
          g\opt(s)
          = \lim_{k\to \infty} \left(  
            \sum_{s' \in \mathcal{K}} \mathbb P^{\pi\opt}_s \left[\tilde{s}_k=s'\right] \cdot  1 
          + \sum_{
          s' \in \mathcal{C}\setminus \mathcal{K}} \mathbb P^{\pi\opt}_s \left[\tilde{s}_k=s'\right]g\opt (s') \right)
          = \lim_{k\to \infty} 
            \sum_{s' \in \mathcal{K}} \mathbb P^{\pi\opt}_s \left[\tilde{s}_k=s'\right]. 
    \]
Therefore, the limiting probability of being in $\mathcal{C}$ is entirely concentrated on $\mathcal{K}$:
\[
\lim_{k \to \infty} \sum_{s' \in \mathcal{C}} \mathbb{P}^{\pi\opt }_s\!\left[\tilde{s}_k = s'\right]
= \sum_{s' \in \mathcal{K}} \lim_{k \to \infty} \mathbb{P}^{\pi\opt }_s\!\left[\tilde{s}_k = s'\right].
\]
By the definition of the probabilistically--safe set 
        \[
            \mathbb{P}^{\hat{\pi}}\bigl[\,\tilde{x}_k(s)\in \mathcal{C},\;\forall k\in\mathbb N\bigr]
            = 
            \lim_{k\to \infty} 
            \sum_{s' \in \mathcal{K}} \mathbb P^{\pi\opt}_s \left[\tilde{s}_k=s'\right] 
            \;\ge\; \alpha.
        \] 
This limiting sum defines the probability of eventually reaching a state in \(\mathcal{C}\), which, by definition, equals \(g\opt(s)\). Since \(g\opt(s) \geq \alpha\), the condition is satisfied, and this completes the proof.
\end{proof}

Intuitively, \cref{theorem:alpha_level} states that the probability of staying safe at confidence level $\alpha$ is connected by the $\alpha$ probability of returning from a current transient state to the safe state in $\mathcal{K}$.

With the theoretical connection between multichain average reward MDPs and probabilistic safety established, we can now leverage the properties of infinite average reward MDPs to formulate this safety problem as a corresponding linear program.
The infinite-horizon average reward problem proposed in \cref{eq:opt_policy} can be solved using the following primal linear program~\citep[Section 9.3]{Puterman2005}:
\begin{equation} \label{eq:lp-multichain-primal}
\begin{aligned}
  \operatorname*{minimize}_{g \in \Real_+^\mathcal{S},\, h \in \Real_+^\mathcal{S}} \quad &
   \sum_{s \in \mathcal{S}} \alpha_s\, g(s)\\
  \operatorname{s.\,t.} \quad &
    g(s) \ge \sum_{s'\in \mathcal{S}} g(s') \, p(s,a,s'), \quad \forall\, s \in \mathcal{S}, a\in \mathcal{A} \\
    & g(s)  \!+\!  h(s) \ge \mathbbm{1}\{s_k \in \mathcal{C} \}\!+\! \sum_{s'\in \mathcal{S}} h(s') \, p(s,a,s'),
\qquad                                                                                                  \forall\, s\in \mathcal{S}, a \in \mathcal{A}
\end{aligned}
\end{equation}
where the value $\alpha_{s_0} \in \Real_+^\mathcal{S}$ may be arbitrary as long as it satisfies that
$
 \alpha_{s_0} \!>\! 0, \; \!\forall\! s \in \mathcal{S},  \hspace{1pt}\sum_{s \in \mathcal{S}} \alpha_s \!=\! 1.
$
To determine the corresponding actions, solving the dual linear program is also necessary. Equation \eqref{eq:main_objective} can be solved using the following dual linear program.
\begin{equation} \label{eq:lp-multichain-dual}
\begin{aligned}
  \operatorname*{maximize}_{z \in \Real_+^{\mathcal{S} \times  \mathcal{A}},\, y \in \Real_+^{\mathcal{S} \times  \mathcal{A}}} \quad &
   \sum_{s \in \mathcal{S}} \sum_{a \in \mathcal{A}} z(s,a) \mathbbm{1}\{s_k \in \mathcal{C} \}, \\
  \operatorname{s.\,t.} \quad &
   \sum_{a \in \mathcal{A}} \!z(s',a) \!-\! \sum_{s\in \mathcal{S}}\! \sum_{a\in \mathcal{A}} \!z(s,a) \, p(s,a,s') \!=\! 0, \forall\, s'\!\in \mathcal{S} \\
  & \sum_{a\in \mathcal{A}} y(s',a) -  \sum_{s\in \mathcal{S}} \sum_{a\in \mathcal{A}} y(s, a) \, p(s,a,s')=   \alpha_{s_0} - \sum_{a \in \mathcal{A}} z(s,a), \; \forall\, s'\in \mathcal{S}.
\end{aligned}
\end{equation}
We construct a policy from any feasible solution to the linear program in \eqref{eq:lp-multichain-dual} for each $s \in \mathcal{S}$ and $a \in \mathcal{A}$ as:
\begin{equation} \label{eq:policy-construction}
  \pi(s) =
  \begin{cases}
     \frac{z(s,a)}{\sum_{a'\in \mathcal{A}} z(s,a)}  &\text{if }\sum_{a'\in \mathcal{A}} z(s,a') > 0,  \\
     \frac{y(s,a)}{\sum_{a' \in \mathcal{A}} y(s,a') }  &\text{otherwise}. \\
  \end{cases}
\end{equation}
It should be mentioned that for each \(s \in \mathcal{S}\), \(z(s,a) > 0\) for exactly one \(a \in \mathcal{A}\), 
and for each \(s \notin \mathcal{S}\), \(y(s,a) > 0\) for exactly one \(a \in \mathcal{A}\), there exists a deterministic policy~\citep{Puterman2005}.

\section{Comparison with Existing Methods}\label{sec:r-work}

In this section, we compare the proposed Average Reward (AVR) approach with three related studies: a robust safety framework based on HJI reachability~\citep{fisac2018general}, the Minimum Discounted Reward (MDR) formulation for reachable sets \citep{akametalu2023minimum}, and the reachability of Markov chains using long-run average rewards~\citep{avila2021reachability}.

HJI-based approaches, formulated in the continuous domain, compute a single, robustly safe set for a deterministic system under worst-case disturbances. These methods numerically solve the HJI PDE using discrete-time Bellman equations, which typically require a specific initialization to converge. The MDR formulation~\citep{akametalu2023minimum}, also continuous, addresses this by introducing a discount factor, which transforms the problem into a contraction mapping that ensures convergence from any initialization. In the MDR method, this discount factor directly influences the size of the final safe set, often represented by the zero-level set of a signed distance function.

In contrast, the AVR method is formulated directly as a discrete MDP with known stochastic transition probabilities. The AVR method employs the average reward criterion to solve the infinite-horizon problem and is formulated without the use of a discount factor. For computation, the AVR method solves the problem using a standard linear program (\cref{eq:lp-multichain-primal,eq:lp-multichain-dual}). Finally, the AVR computation yields the optimal gain function, which provides a probabilistic safety function (\cref{theorem:alpha_level}) where different level sets correspond to different safety confidence levels.

Although our proposed and reachability for Markov chains methods both use the average reward criterion, the AVR method is formulated to solve the indefinite probabilistic safety problem. This objective differs from the standard reach-avoid task. While the reachability framework is applied to find the probability of reaching a set before avoiding, we aim to find the probability of remaining in the safe constraint set indefinitely (that is, the probability of never reaching the unsafe set). This difference in objective is reflected in the construction of the transition function. To solve the reach-avoid problem, the Avila and Junca method constructs a modified transition kernel where both the target set and the avoid set are made absorbing. In contrast, our method constructs a transition function (\cref{eq:transition-function}) where only the unsafe set is made absorbing. This specific construction of the transition model enables the optimal gain function to represent the probability of indefinite safety, allowing us to formally prove in \cref{theorem:alpha_level} that the optimal gain function is directly equal to this maximum probability.

\section{Numerical Validation}
\label{sec:N_results}

We evaluate our approach numerically on two benchmark control systems: a double integrator and an inverted pendulum. The simulations were implemented in  Julia, using the JuMP modeling language and the MOSEK optimization solver. Our AVR method is solved as a linear program, whereas the MDR method is implemented for comparison and solved using value iteration (VI). All computations were executed on a computer equipped with a 64-core AMD Ryzen Threadripper 3970X CPU and 256 GB of RAM.

The continuous state and action spaces of the two benchmark control systems are discretized into uniform grids. The stochastic transition probabilities are constructed using a Monte Carlo approach: for each discrete state-action pair, 100 simulations are run with disturbances from the known distribution and bounds. Each resulting next state in the continuous state space is mapped to the closest discretized state on the grid, identified using a k-nearest-neighbors search (KDTree and knn)\footnote{Here, KNN serves as a non-parametric function approximator, where the value of a state is estimated from the values of its nearest neighbors on the discretized grid.}. The transition probabilities are thus computed from the frequency of these assignments.

The discrete‑time Double Integrator dynamics are considered as 
\[
 \mathrm{x}_{k+1} = \mathrm{x}_{k} + \mathrm{v}_{k}\,\Delta t,\quad
\mathrm{v}_{k+1} = \mathrm{v}_{k} + \bigl({u}_{k}+\tilde{d}_k  \hspace{1pt}\bigr)\,\Delta t.
\]
$\mathrm{x}$ is position and $\mathrm{v}$ is velocity. The state space is bounded by $\mathcal{S}=\{(\mathrm{x},\mathrm{v})|-1\le \mathrm{x}\le5,-5\le \mathrm{v}\le5\}$. The disturbance $\tilde{d}_k$ is drawn from a Normal distribution $\mathcal{N}(0, 1)$, clamped to $[-1,1]$. The control input $u_{k}$ is in $[-2,2]$, and the constraint set is $\mathcal{C} = \{(\mathrm{x},\mathrm{v})\in\mathbb{R}^{2}|0\le \mathrm{x}\le4,-3\le \mathrm{v}\le3\}$. The time step is $\Delta t=0.1$. For this specific system, we use a $161 \times 161$ state grid and 81 actions.

The discrete-time dynamics of a stochastic, nonlinear Inverted Pendulum are:
\[\theta_{k+1}=\theta_{k}+\omega_{k}\Delta t, \qquad
\omega_{k+1}=\omega_{k}+(\frac{g_{p}}{l_{p}}sin(\theta_{k})+\frac{1}{m_{p}l_{p}^{2}}u_{k}+\tilde{d_{k}})\Delta t\]
Here, the state consists of the angle $\theta_{k}$ and the angular velocity $\omega_{k}$, bounded by $\mathcal{S}=\{(\theta,\omega)|-0.5\le\theta\le0.5, -1.0\le\omega\le1.0\}$. The disturbance $\tilde{d}_{k}$ is drawn from a Normal distribution $\mathcal{N}(0,1)$, clamped to $[-0.75,0.75]$. The control input $u_{k}$ is in $[-3,3]$, and the constraint set is $\mathcal{C}=\{(\theta,\omega)\in\mathbb{R}^{2}|-0.3\le\theta\le 0.3, -0.6\le\omega\le0.6\}$. The time step is $\Delta t=0.1$. For this system, we use a $201 \times 201$ state grid and 81 discrete actions.

\begin{figure}
\centering
    % --- First Figure ---
    \begin{minipage}[b]{0.5\textwidth}
    \flushright
        \includegraphics[height=4.5cm,width=6cm]{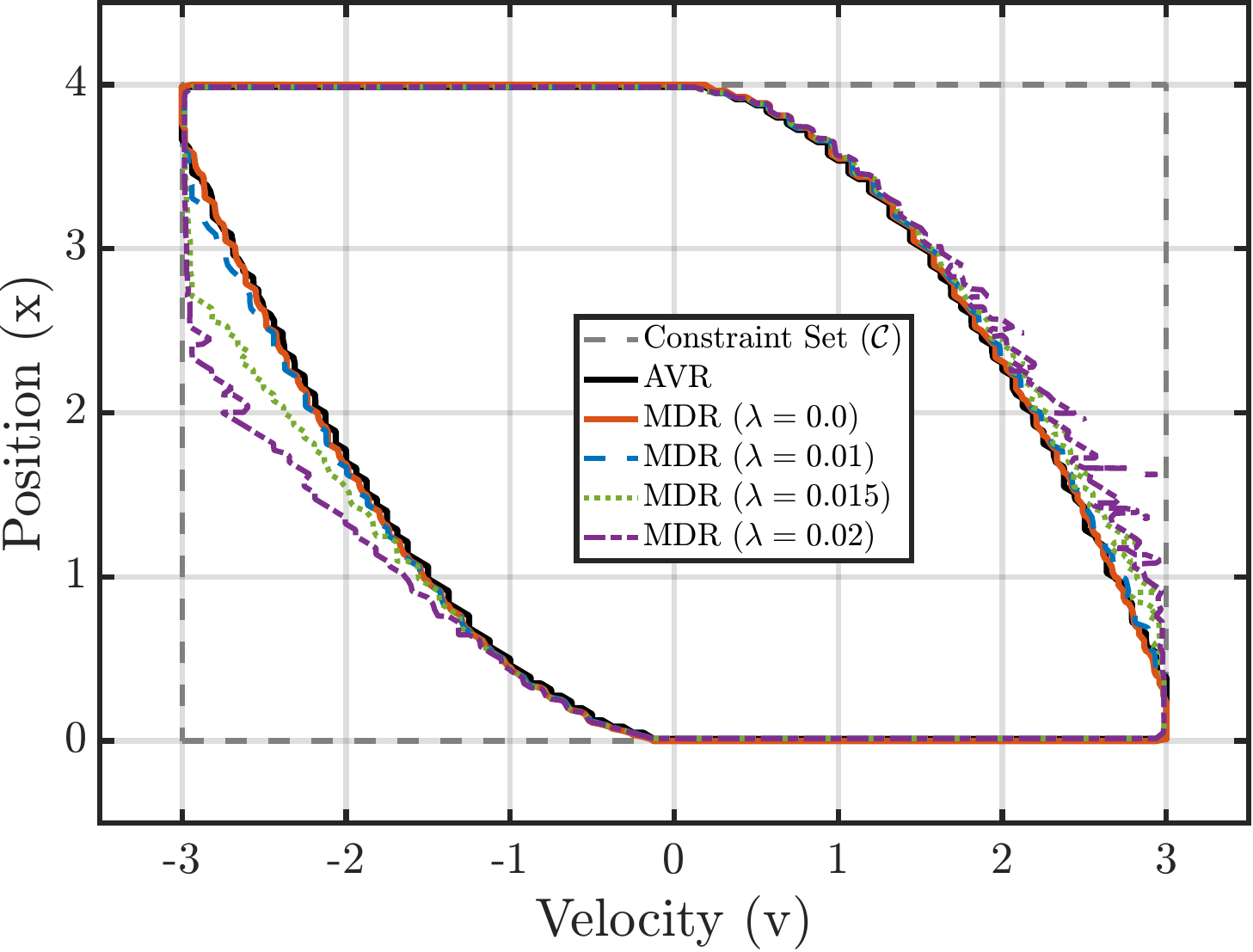}
    \end{minipage}%
    % --- Second Figure ---
    \begin{minipage}[b]{0.5\textwidth}
    \centering
        \includegraphics[height=4.5cm, width=6cm]{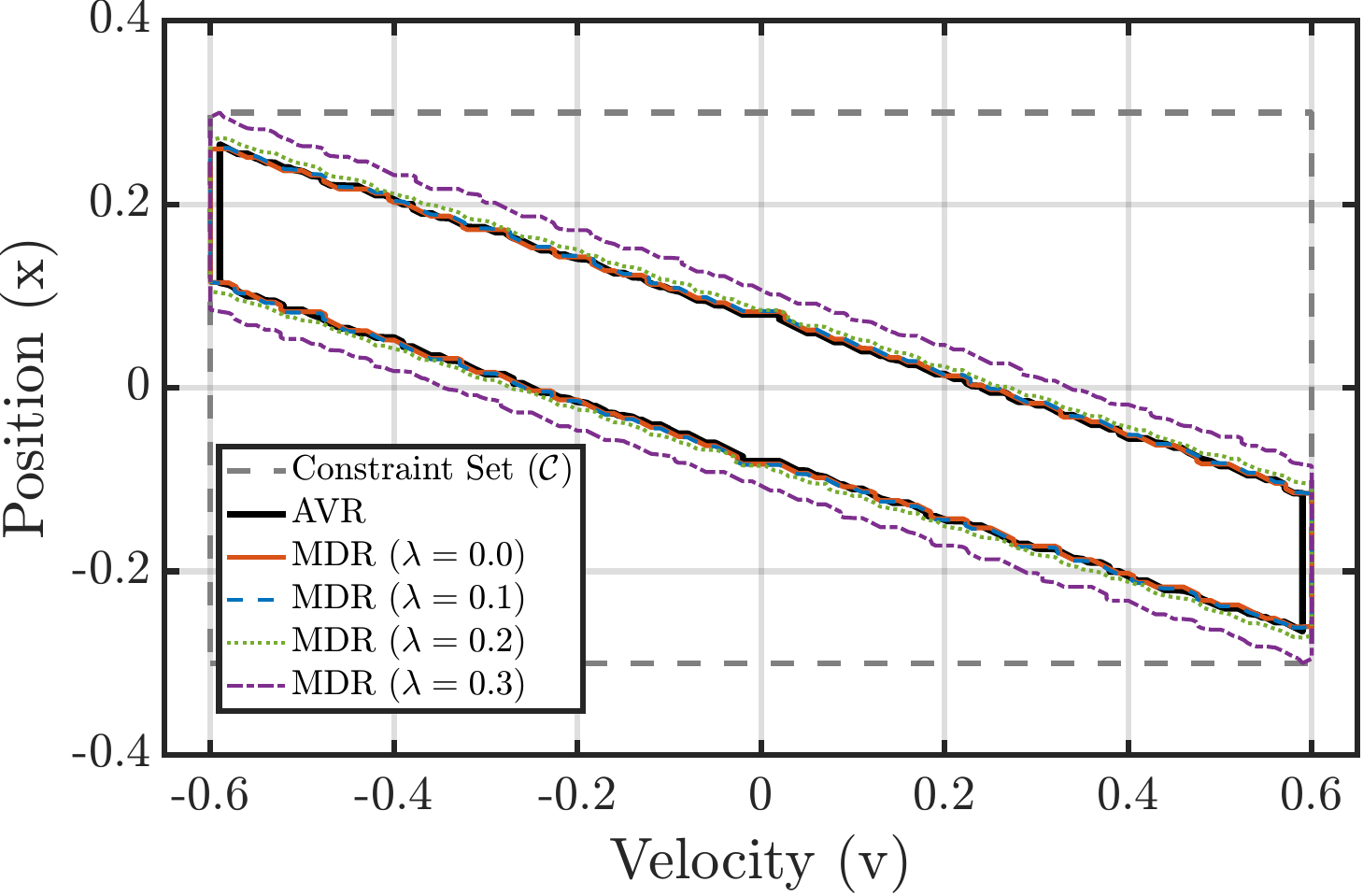}
    \end{minipage}

    \caption{Safe sets computed by AVR and MDR for the Double Integrator (right) and Inverted Pendulum (left). The dashed gray line outlines the constraint set $\mathcal{C}$, the solid black line shows AVR's safe set, and the colored lines show MDR's safe set ($Z(\bm{x})=0$) for different $\lambda$.}
    \label{fig:comparison_plots}
\end{figure}
\cref{fig:comparison_plots} illustrates the computed safe sets for the Double Integrator (left) and Inverted Pendulum (right). The figure shows the safe set boundary from our AVR method (solid black line) alongside the boundaries from the MDR method using various discount rates ($\lambda$). The boundaries from the MDR method are observed to change as $\lambda$ is varied. The AVR result aligns with the undiscounted ($\lambda=0$) MDR case, considered the analytical solution~\citep{akametalu2023minimum}.

\begin{figure}
    \centering
    % --- First Figure ---
    \begin{minipage}[b]{0.5\textwidth}
        \flushright
        \includegraphics[height=4.5cm]{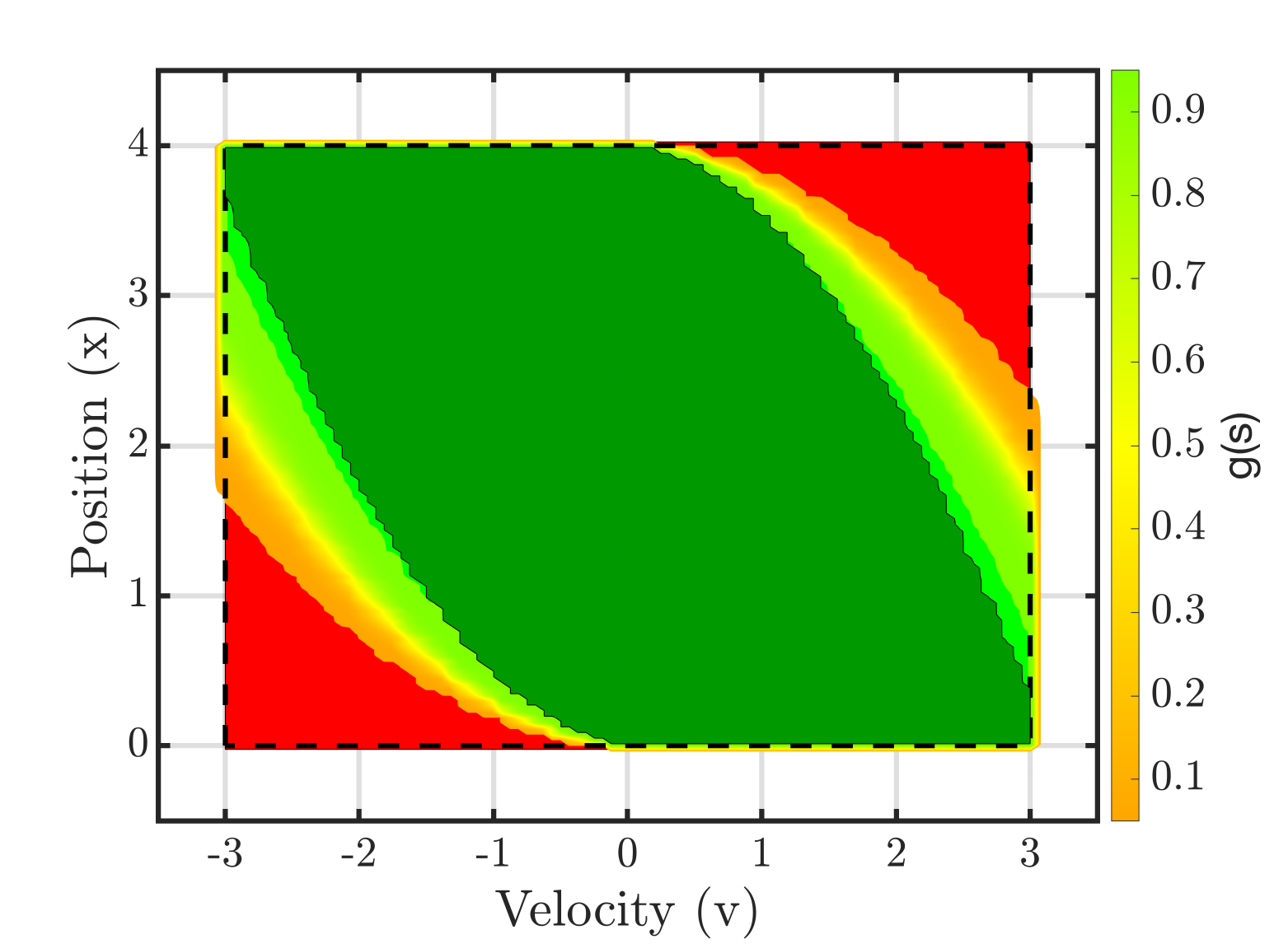}
    \end{minipage}%
    \hfill
    % --- Second Figure ---
    \begin{minipage}[b]{0.5\textwidth}
        \centering
        \includegraphics[height=4.5cm]{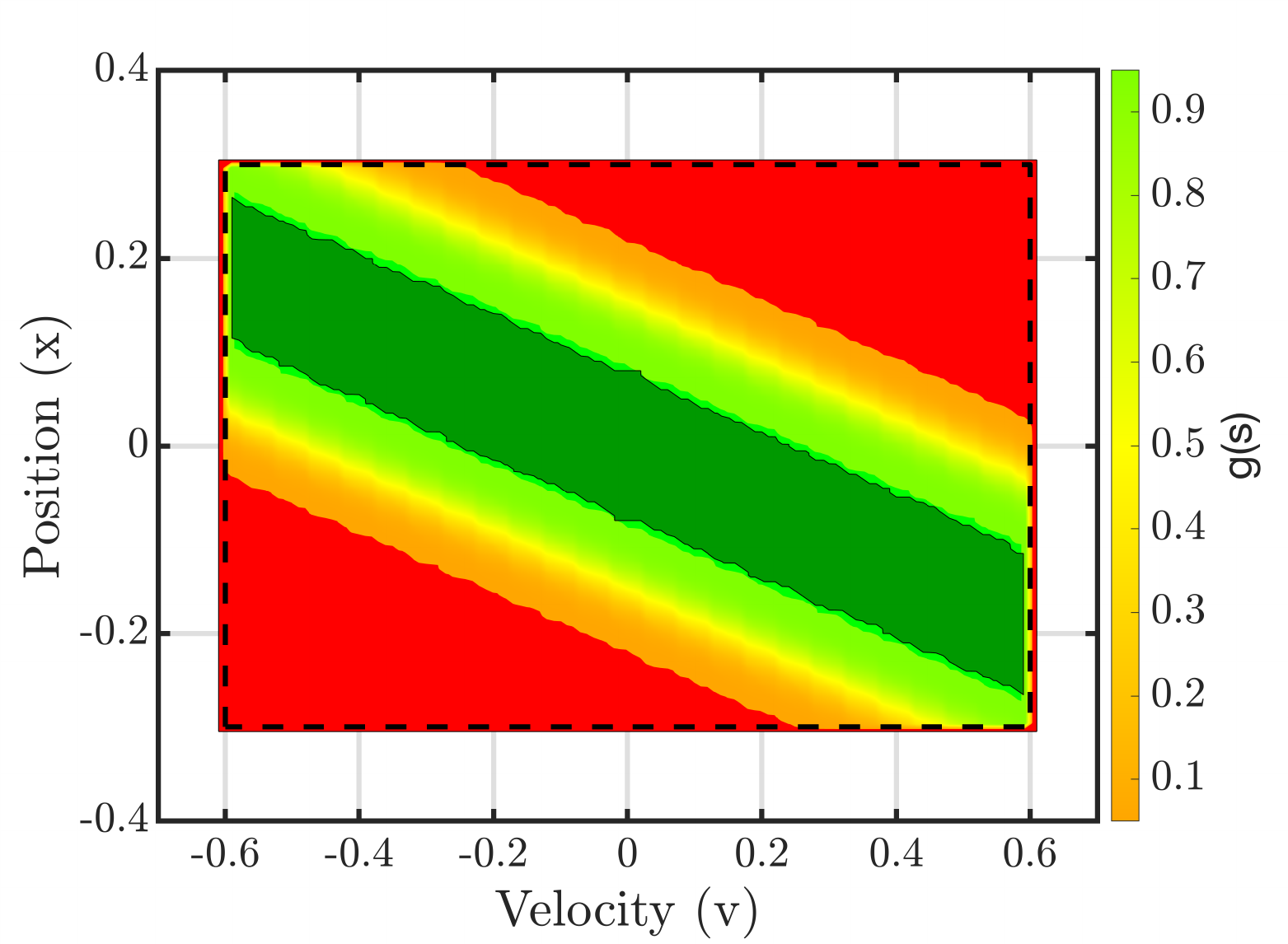}
    \end{minipage}
    \caption{Safe sets computed using AVR as level sets of $g(s)$ for (a) Double Integrator and (b) Inverted Pendulum systems. The set $\mathcal{K}$ is dark green.}
    %The red region corresponds to the zero-level set of $g(s)$, whereas the dark green region indicates the one-level set of $g(s)$, denoted as $\mathcal{K}$.}
    \label{fig:g_plots}
\end{figure}
\cref{fig:g_plots} presents the numerical results for the proposed AVR method, visualizing the optimal gain function, $g\opt(s)$, as a safety distribution for the Double Integrator (right) and Inverted Pendulum (left). In each subplot, the dashed black lines are the constraint set, $\mathcal{C}$, within which the color map indicates the level of probabilistic safety. The dark green area represents the safe set, $\mathcal{K}$, where the optimal gain $g\opt(s)=1$. Conversely, the red region corresponds to the zero-level set where $g\opt(s)=0$, representing unsafe states. The gradient from orange to light green illustrates the transient states between zero-level set and one-level set of $g\opt(s)$, showing states with varying probabilities of remaining safe as a function of the gain value.

\begin{figure}
    \centering
\includegraphics[height=4.5cm]{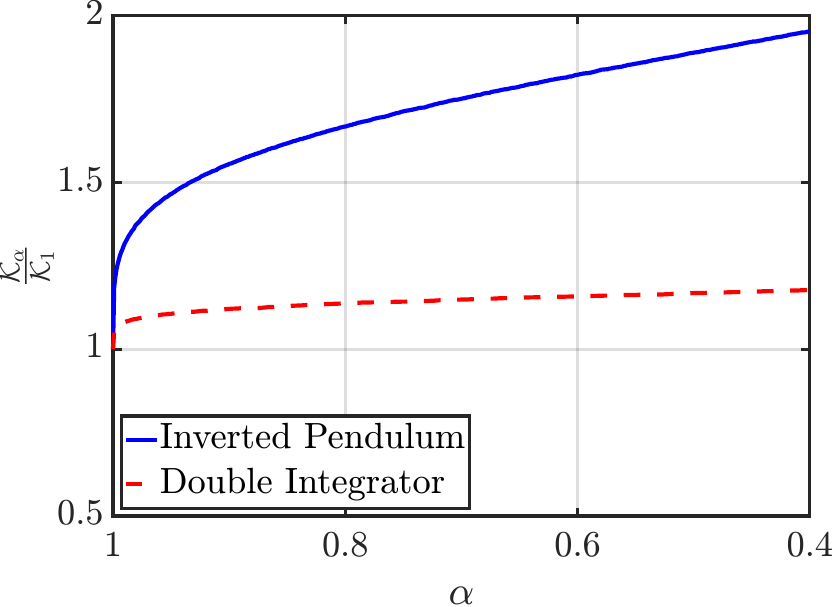}
\caption{The relative size of the probabilistically-safe set $\mathcal{K}_{\alpha}$ as a function of the safety confidence level $\alpha$. The y-axis plots the ratio of the size of $\mathcal{K}_{\alpha}$ to the 100\% safe set, $\mathcal{K}$. Results are shown for the Double Integrator (dashed red) and Inverted Pendulum (solid blue).}
    \label{fig:ratio}
\end{figure}
\cref{fig:ratio} plots the relative size of the probabilistically-safe set $\mathcal{K}_{\alpha}$ as a function of the safety confidence level $\alpha$. The y-axis shows the ratio of the size of the $\alpha$-safe set to the size of $\mathcal{K}$. Both systems start at a ratio of 1 when $\alpha=1.0$. For the Double Integrator (dashed red line), the curve remains relatively flat, increasing only slightly as $\alpha$ decreases. In contrast, the Inverted Pendulum (solid blue line) shows a significant and steep increase, with the safe set's relative size approaching 2.0 as $\alpha$ moves toward 0.4.

\begin{figure}
    \centering
    % --- First Figure ---
    \begin{minipage}[b]{0.5\textwidth}
        \flushright
        \includegraphics[height=4.5cm]{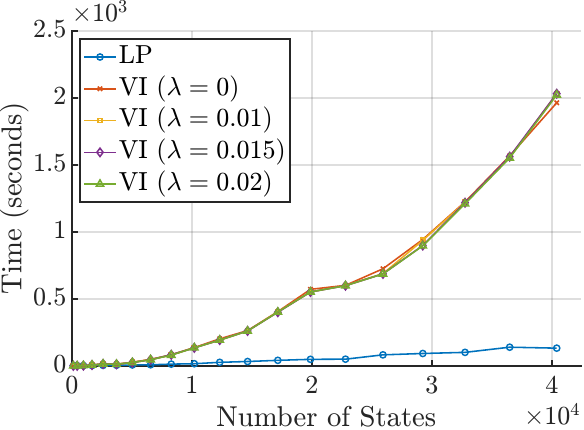}
    \end{minipage}%
    \hfill
    % --- Second Figure ---
    \begin{minipage}[b]{0.5\textwidth}
        \centering
        \includegraphics[height=4.5cm]{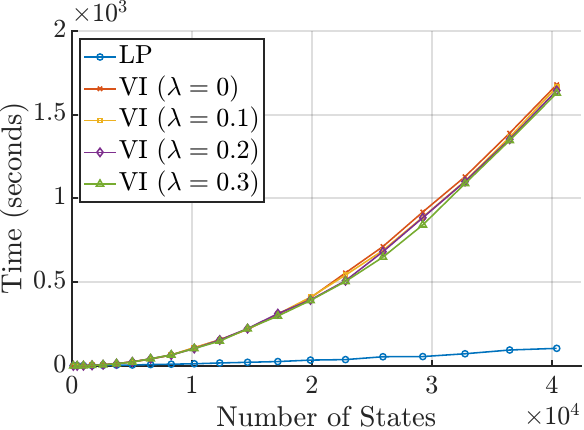}
    \end{minipage}
    \caption{Runtime of the AVR (LP) and the MDR (VI) method for the Inverted Pendulum (left) and the Double Integrator (right) as a function of the total number of discrete states.
    }
    \label{fig:time_plot}
\end{figure}
\cref{fig:time_plot} compares the computational time for the AVR (LP) method and the MDR (VI) method on the Inverted Pendulum (left) and Double Integrator (right) systems. The plots show that the computation time for the AVR method, which is solved as a single linear program, scales much more efficiently with the number of states. In contrast, the time for the value iteration grows substantially faster. The plots also indicate that the value iteration computation time is largely insensitive to the specific discount rate chosen, as the lines for all $\lambda$ values are close to each other.
\section{Conclusion and Future Work}
\label{sec:Conc}
This study introduces a framework for probabilistic safety analysis in stochastic control systems by applying the average-reward MDP properties. A fundamental link was established between the probabilistic safe set $\mathcal{K}_{\alpha}$ and the optimal gain function $g\opt(s)$, which serves as the probabilistic safety value function. This approach eliminates dependence on the arbitrary discount factor, which is critical in MDR and other discounted methods, a factor known to affect both the convergence rate and the size of the safe set. By formulating the safety objective as an infinite-horizon average reward problem, the probabilistic safety value function can be computed efficiently by solving a standard linear program, which is significantly faster than value iteration (VI) methods typically used for discounted formulations. Numerical results on the Double Integrator and Inverted Pendulum systems confirm that the AVR method yields accurate and safe set boundaries.

Future research focuses on extending this framework to move beyond the computation of safe policies to selecting and utilizing only those safe policies that also satisfy predefined optimality conditions. Specifically, the objective will be to choose a subset of probabilistically safe policies that maximize performance metrics, thereby coupling safety guarantees with a desired level of system optimality. 

\newpage

\bibliography{Ref}

\appendix

\section{Lemmas for Theorem 3}
\label{app:proof-thm3}

\begin{lemma}\label{lem:transient_state}
  Each $s \in \mathcal{C} \setminus \mathcal{K}$ is transient:
    \[
   \lim_{k \to \infty} \mathbb{P}^{\pi} _s\!\left[\tilde{s}_k = s \right] = 0, \quad \forall \pi \in \PiSD . 
  \]
\end{lemma}
\begin{proof}
Suppose, for the sake of deriving a contradiction, that $s \in  \mathcal{C} \setminus \mathcal{K}$ is considered as \emph{recurrent} if and only if the following summation of time step transition probability is: 
\[
\sum_{k=0}^\infty p^{k}(s|s) = \infty \quad \forall s \in \mathcal{C} \setminus \mathcal{K}
\]
where 
\[
\sum_{k=0}^{\infty}p^{k}(s|s) = 1 + p(s|s) + p^2(s|s) + p^3(s|s) + \dots + p^\infty(s|s).\]
Then, if and only if the
limiting probability equals 
\(p(s|s)^\infty = 
   \lim_{k \to \infty} \mathbb{P}^{\pi} _s\!\left[\tilde{s}_k = s \right] = 1,\forall \pi \in \PiSD . 
  \)
Whereas this contradicts the fact that the state is transient if and only if
\[
\sum_{k=0}^\infty p^{k}(s|s) <\infty \quad \forall s \in \mathcal{C} \setminus \mathcal{K}.
\]
Hence the sum must be finite and
\(
p(s|s)^\infty  =
\lim_{k \to \infty} \mathbb{P}^{\hat{\pi}}_s[\tilde{s}_k= s] = 0 \quad \forall s \in \mathcal{C} \setminus \mathcal{K}. 
\)
Therefore, $s \in \mathcal{C} \setminus \mathcal{K}$ is \emph{transient}.
\end{proof}

\begin{lemma} 
For all $s \in \mathcal{S} \setminus \mathcal{C}$,
\label{lem:gain-failed}
  \[
    g\opt(s) = 0.
  \]
\end{lemma}

\begin{proof}
This result follows directly from the structure of the transition probabilities defined in \eqref{eq:transition-function}. States outside the constraint set $\mathcal{C}$ transition only to themselves, which are classified as unsafe and do not yield any reward.
\end{proof}

\begin{lemma}
\label{lem:gain-one-safe}
For all \( s \in \mathcal{S} \), we have
\[
  s \in \mathcal{K} \;\Longleftrightarrow\; g^{\star}(s) = 1.
\]
\end{lemma}

\begin{proof}
\medskip
\noindent
First, we show \(s\in \mathcal{K} \implies g\opt(s)=1\), there exists a deterministic stationary
        policy \(\pi\) such that
        \[
          \P^{\pi}_s \left[  \tilde{x}_k \in \mathcal{C}\right] = 1\quad\text{a.s.\ for every }k\in\mathbb N.
        \]
Under this policy 
        \(\mathbbm{1}\{\tilde{x}_k \in \mathcal{C}\} = 1\) for all \(k\).
Hence: \[g^{\pi}(s) = \limsup_{N\to\infty}\frac{1}{N}\sum_{k=0}^{N-1} 1 = 1.\]
 Since $r(s,a,s') \leq 1, \forall s,s' \in \mathcal{S}, a \in \mathcal{A}$, 
\(
g^{\pi} \le g^{\star}(s)\le 1 .
\)

Second, we start $  s \notin \mathcal{K} \;\implies\; g\opt(s) < 1$ (since $g\opt(s) \le 1$).
Since $s\notin \mathcal{K}$, we have that $\forall \pi \in \PiSD$, $\exists \zeta > 0$, such that
  \[
    \P^{\pi}_s \left[ \tilde{x}_{\zeta} \notin \mathcal C \right] > 0 ,
    \qquad
    \P^{\pi}_s \left[ \tilde{x}_k \notin \mathcal C \right] = 0, \quad\forall k < \zeta .
  \]
By the construction of transition probabilities in \eqref{eq:transition-function}: 
  \[
    \P^{\pi}_s \left[ \tilde{s}_{\zeta} \notin \mathcal C \right] = \P^{\pi}_s \left[ \tilde{x}_{\zeta} \notin \mathcal C \right] > 0 ,
    \qquad
    \P^{\pi}_s \left[ \tilde{s}_k \notin \mathcal C \right] = 0, \quad\forall k < \zeta .
  \]
According to 
multichain optimality equation~\citep{Puterman2005}: 
\[
    \max_{a \in \mathcal{A}} \left\{ \sum_{s' \in \mathcal{S}} p(s, a, s') g^\pi(s') - g^\pi(s) \right\} = 0,
\]
for a fixed policy and specific probability set up, the gain is:  
\[ g^{\pi}(s) = \sum_{s'\in \mathcal{S}} \P^{\pi}_{s} \left[ \tilde{s}_{\zeta} = s' \right] \cdot g^{\pi}(s'),
    \quad
    \forall \pi\in \PiSD. \]
We know from \cref{lem:gain-failed} that for each $s' \in  \mathcal{S} \setminus \mathcal{C}$ \( g^{\pi}(s') = 0, \hspace{2pt}\forall \pi \in \PiSD.\)
Then:
  \[
    g^{\pi}(s) = \sum_{s'\in \mathcal{S}} \P^{\pi}_{s} \left[ \tilde{s}_{\zeta} = s' \right] \cdot g^{\pi}(s')
    \le \left(1- \P^{\pi}_s \left[ \tilde{s}_{\zeta}  \in  \mathcal{S} \setminus \mathcal{C} \right] \right)
    < 1,
    \quad
    \forall \pi\in \PiSD.
  \]
  Therefore, from the existence of an optimal policy \(
   g\opt(s) \le 1.\)
\medskip
\noindent
Since both directions hold, we conclude
\(s\in \mathcal{K} \;\Longleftrightarrow\; g^{\star}(s)=1.\)
\end{proof}
With reference to \cref{lem:gain-failed}, the following theorem shows that any state with a gain equal to 1 is considered a safe state, and the union of these states forms the safe set. Intuitively, the theorem implies that safety is ensured if there exists a policy such that the expected average reward, defined by the indicator function of the constraint set, remains 1 at all time steps.
\end{document}